\documentclass[11pt]{article}

\usepackage{fullpage}
\usepackage[numbers]{natbib}
\usepackage[dvipsnames]{xcolor}
\usepackage[linktocpage]{hyperref}
\hypersetup{
    unicode=false,          % non-Latin characters in Acrobatӳ bookmarks
    colorlinks=true,        % false: boxed links; true: colored links
    linkcolor=red,          % color of internal links (change box color with linkbordercolor)
    citecolor=OliveGreen,        % color of links to bibliography
    filecolor=magenta,      % color of file links
    urlcolor=cyan           % color of external links
}

\usepackage{algorithm}
\usepackage[noend]{algpseudocode}
\usepackage{amsmath,amssymb,amsthm}
\usepackage{bm}
\usepackage{bbm}
\usepackage{bbold}
\usepackage{caption}
\usepackage{enumerate}
\usepackage{xspace}
\usepackage{cleveref}
\usepackage{thm-restate}
\usepackage{graphicx}
\usepackage{subcaption}

\global\long\def\congest{\mathrm{con}}%
\newcommand{\load}{\text{load}}
\newcommand{\sep}{\operatorname{sep}}

\newcommand{\spa}{\operatorname{spars}}
\newcommand{\congslack}{\Theta_{\ref{thm: linked ED}}}

\newcommand{\conConst}{\Theta_{\ref{thm:flow character}}}

\author{
\begin{tabular}[t]{c@{\extracolsep{2.5em}}cc} 
        Bernhard Haeupler  &  D Ellis Hershkowitz & Zihan Tan \\
        \small ETH Z\"urich \&  & \small  ETH Z\"urich & \small DIMACS, \\
        \small  Carnegie Mellon University &  & \small Rutgers University \\
\end{tabular}
}

\newtheorem{theorem}{Theorem}
\newtheorem{claim}{Claim}

\newtheorem{lemma}{Lemma}
\newtheorem{definition}{Definition}

\newcounter{note}[section]

\newcommand{\set}[1]{\left\{ #1 \right\}}

\title{Parallel Greedy Spanners}

\date{}

\begin{document}

\maketitle

% ---- TODO FOR FINAL VERSION ----
% Figures

% Beef up motivation / Bernhard comments / ref expander decomps and such
% Mention works for weighted setting with extra log n (assuming poly aspect ratio)
% Incorporate Greg refs / mine for applications
% Change "k-crucial" to "k-spanned"?
% Add hop-bounded spanner proof with tradeoffs
% Add a conclusion; open questions: 

% ---- OLD TODOs ----
% Make work for weighted case
% Make algorithmic; try and get new algorithmic result that can sell e.g. dynamic or parallel
% Generalize matching to low arboricity and potentially further
% Add figures
% Add lower bound (ack Bodwin)
% Add proof of general length-constrained expander decomps

\begin{abstract}

A $t$-spanner of a graph is a subgraph that $t$-approximates pairwise distances. The greedy algorithm is one of the simplest and most well-studied algorithms for constructing a sparse spanner: it computes a $t$-spanner with $n^{1+O(1/t)}$ edges by repeatedly choosing any edge which does not close a cycle of chosen edges with $t+1$ or fewer edges.

We demonstrate that the greedy algorithm computes a $t$-spanner with $t^3\cdot \log^3 n \cdot n^{1 + O(1/t)}$ edges \emph{even when a matching of such edges are added in parallel}. In particular, it suffices to repeatedly add any matching where each individual edge does not close a cycle with $t +1$ or fewer edges but where adding the entire matching might. Our analysis makes use of and illustrates the power of new advances in length-constrained expander decompositions.
\end{abstract}
\thispagestyle{empty}

\newpage
\setcounter{page}{1}

\newcommand{\pg}{\textnormal{\textsf{pg}}\xspace}
\newcommand{\dist}{d}
\newcommand{\vol}{\mathsf{vol}}

\section{Introduction}
A spanner of a graph is a sparse subgraph that approximately preserves pairwise distances. 

\begin{definition}[$t$-Spanner]
Given graph $G = (V, E)$ and $t \geq 1$, a $t$-spanner is a subgraph $H = (V, E')$ of $G$ such that for every $u,v \in V$ we have
\begin{align*}
    d_H(u,v) \leq t \cdot d_G(u,v).
\end{align*}
\end{definition}
\noindent (See \Cref{sec:notation} for standard notation and definitions). Since their formalization by \cite{peleg1989graph}, spanners have become indispensable in graph sparsfication and graph algorithms. For example, they have found applications in distributed broadcast \cite{awerbuch1990cost, awerbuch1992efficient}, network synchronization \cite{awerbuch1985communication,peleg1987optimal,awerbuch1990cost,awerbuch1992efficient,peleg2000distributed}, overlay, sensor and wireless networks \cite{braynard2002opus,vogel2003priority,kostic2002latency, von2004gathering,ben2004splast,shpungin2010near}, VLSI circuit design \cite{cong1991performance,cong1992provably,salman2001approximating}, routing \cite{wu2002light,thorup2001compact}, distance oracles \cite{thorup2005approximate, peleg2000proximity, roditty2005deterministic} and approximate shortest paths \cite{cohen1998fast,roditty2004dynamic,elkin2005computing,elkin2004efficient,feigenbaum2005graph}. Generally speaking, these works make use of sparse spanners; namely, spanners with a small number of edges.

One of the simplest and most well-studied ways of computing a sparse $t$-spanner $H$ is the greedy algorithm. The greedy algorithm is based on what we will call \emph{unspanned edges}.
\begin{definition}[Unspanned Edge]
    Given graph $G = (V,E)$ and subgraph $H\subseteq G$, say that edge $\{u,v\} \in E$ is $t$-unspanned with respect to $H \subseteq G$ if $d_H(u,v) > t$.
\end{definition}

% for $G = (V,E)$ is as follows: initialize $H$ to the empty subgraph; while there exists an edge $e\in E$ that is $t$-unspanned with respect to $H$, add $e$ to $H$; return $H$.
\noindent In other words, $\{u,v\}$ is $t$-unspanned with respect to $H$ if it does not close a $(t+1)$-cycle in $H$. To compute a $t$-spanner $H$, the greedy algorithm simply repeatedly adds $t$-unspanned edges to $H$ until none exist; see \Cref{alg:seqGreedy}.

% \noindent \begin{minipage}{\linewidth}
\begin{algorithm}[H]
    \caption{Sequential Greedy Algorithm for $t$-Spanners}
    \label{alg:seqGreedy}
    \begin{algorithmic}[0] % The number tells where the line numbering should start
            \State \textbf{Input:} Graph $G = (V,E)$ and $t \geq 1$
            \State \textbf{Output:} A $t$-spanner $H$ of $G$
            \State Initialize $H \gets \emptyset$
            \While{$\exists e \in E$ that is $t$-unspanned with respect to $H$}:
                \State $H \gets H \cup \{e\}$
            \EndWhile
            \State \Return $H$
    \end{algorithmic}
\end{algorithm}
% \end{minipage}
% \vspace{2mm}

\noindent It is easy to see that the resulting $H$ is indeed a $t$-spanner of $G$ since, by construction, every edge in $G$ has a length at most $t$ path in $H$ between its endpoints. Additionally, a classic analysis shows the result is sparse, containing at most most $n^{1 + O(1/t)}$ edges \cite{althofer1993sparse, ahmed2020graph}: by construction the graph output has girth at least $t$ (i.e.\ contains no cycles of length $t$ or less) and classic so-called Moore bounds state that a graph with girth at least $t$ contains at most $n^{1 + O(1/t)}$ edges. Assuming the ``Erdős Girth Conjecture'' \cite{erdos1965some}, this sparsity is asymptotically-optimal. Furthermore, the greedy spanner is existentially-optimal in that its output is as sparse as possible for any given graph family \cite{filtser2016greedy}.

However, one downside of the greedy algorithm is its apparent sequentialness: whether an edge is $t$-unspanned with respect to $H$ is dependent on what edges have previously been added to $H$. Furthermore, the sparsity analysis of the output spanner is also quite delicate in that it relies on the output spanner having girth at most $t$ and has little to say if the girth is much smaller than $t$.

In summary, the (sequential) greedy algorithm is one of the most well-studied algorithms for constructing a sparse $t$-spanner but seems inherently sequential and has a delicate sparsity analysis.

\subsection{Our Results}
In this work, we show that the classic sequential greedy algorithm, in fact, admits significant opportunities for parallelism while (approximately) retaining its sparsity guarantees. In particular, we show that even when many unspanned edges are added at once the resulting graph is still sparse. To formalize this, we generalize the above notion of an unspanned edge to a set of edges as follows.
\begin{definition}[Unspanned Edge Set]
    Given graph $G = (V,E)$ and subgraph $H\subseteq G$, edges $\hat{E} \subseteq E$ are \emph{$t$-unspanned} if each $e \in \hat{E}$ is $t$-unspanned with respect to $H$.
\end{definition}
\noindent Then, the parallel greedy algorithm computes a $t$-spanner $H$ by simply repeatedly adding matchings of $t$-unspanned edges to $H$; see \Cref{alg:parGreedy} and \Cref{fig:pg}.
% for computing a $t$-spanner of $G = (V,E)$ as follows: initialize $H$ to the empty subgraph; while there exists a matching $M \subseteq E$ that is $t$-unspanned with respect to $H$, add $M$ to $H$; return $H$.

% \noindent \begin{minipage}{\linewidth}
\begin{algorithm}[H]
    \caption{Parallel Greedy Algorithm for $t$-Spanners}
    \label{alg:parGreedy}
    \begin{algorithmic}[0] % The number tells where the line numbering should start
            \State \textbf{Input:} Graph $G = (V,E)$ and $t \geq 1$
            \State \textbf{Output:} A $t$-spanner $H$ of $G$
            \State Initialize $H \gets \emptyset$
            \While{$\exists e \in E$ that is $t$-unspanned}:
                \State Let $M \subseteq E$ be any matching in $G$ that is $t$-unspanned with respect to $H$ 
                \State $H \gets H \cup M$
            \EndWhile
            \State \Return $H$
    \end{algorithmic}
\end{algorithm}
% \end{minipage}
% \vspace{2mm}

\begin{figure}[h]
    \centering
    \begin{subfigure}[b]{0.24\textwidth}
        \centering
        \includegraphics[width=\textwidth,trim=0mm 0mm 260mm 0mm, clip]{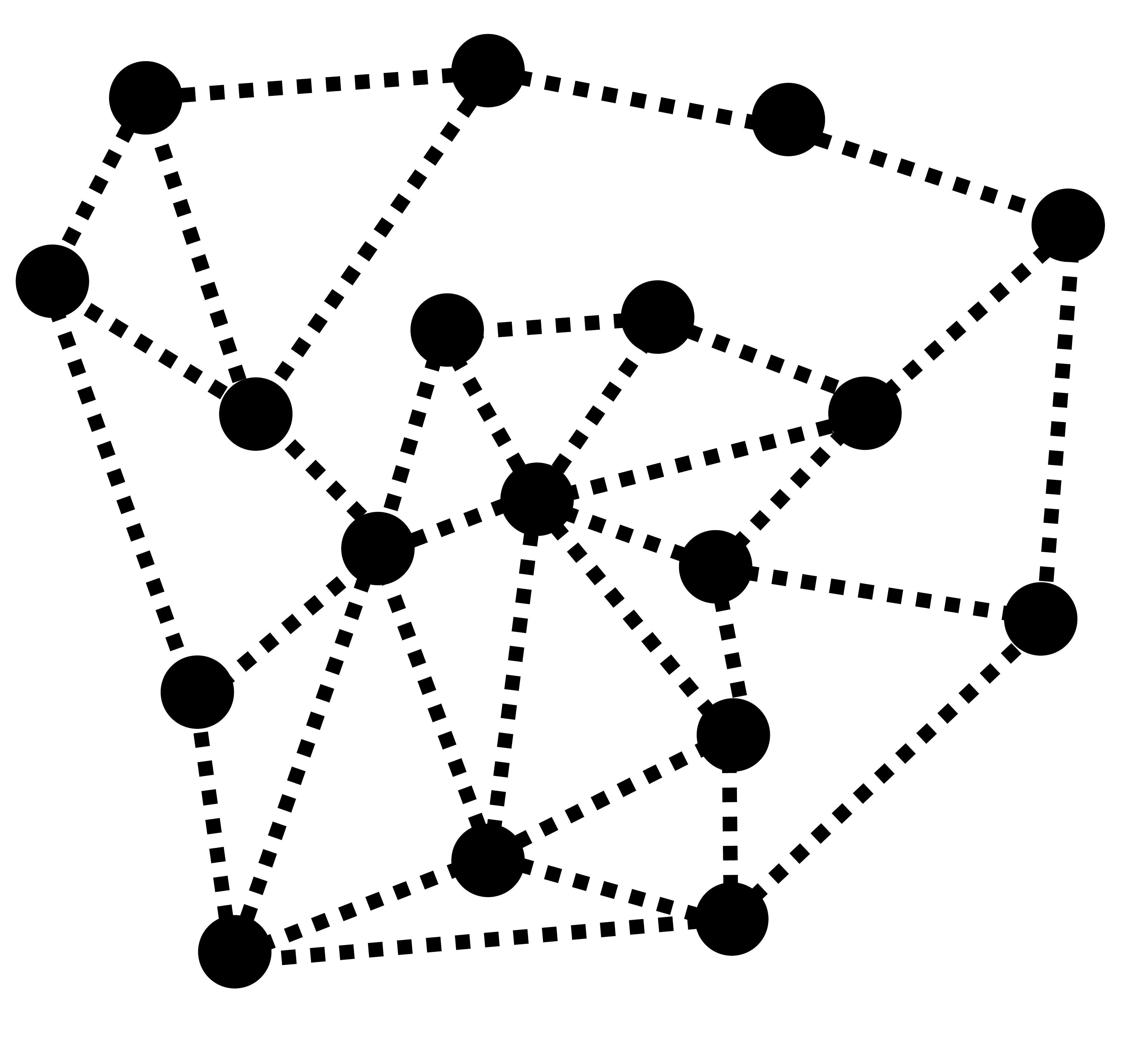}
        \caption{Input $G$.}\label{sfig:pg1}
    \end{subfigure}    \hfill
    \begin{subfigure}[b]{0.24\textwidth}
        \centering
        \includegraphics[width=\textwidth,trim=0mm 0mm 260mm 0mm, clip]{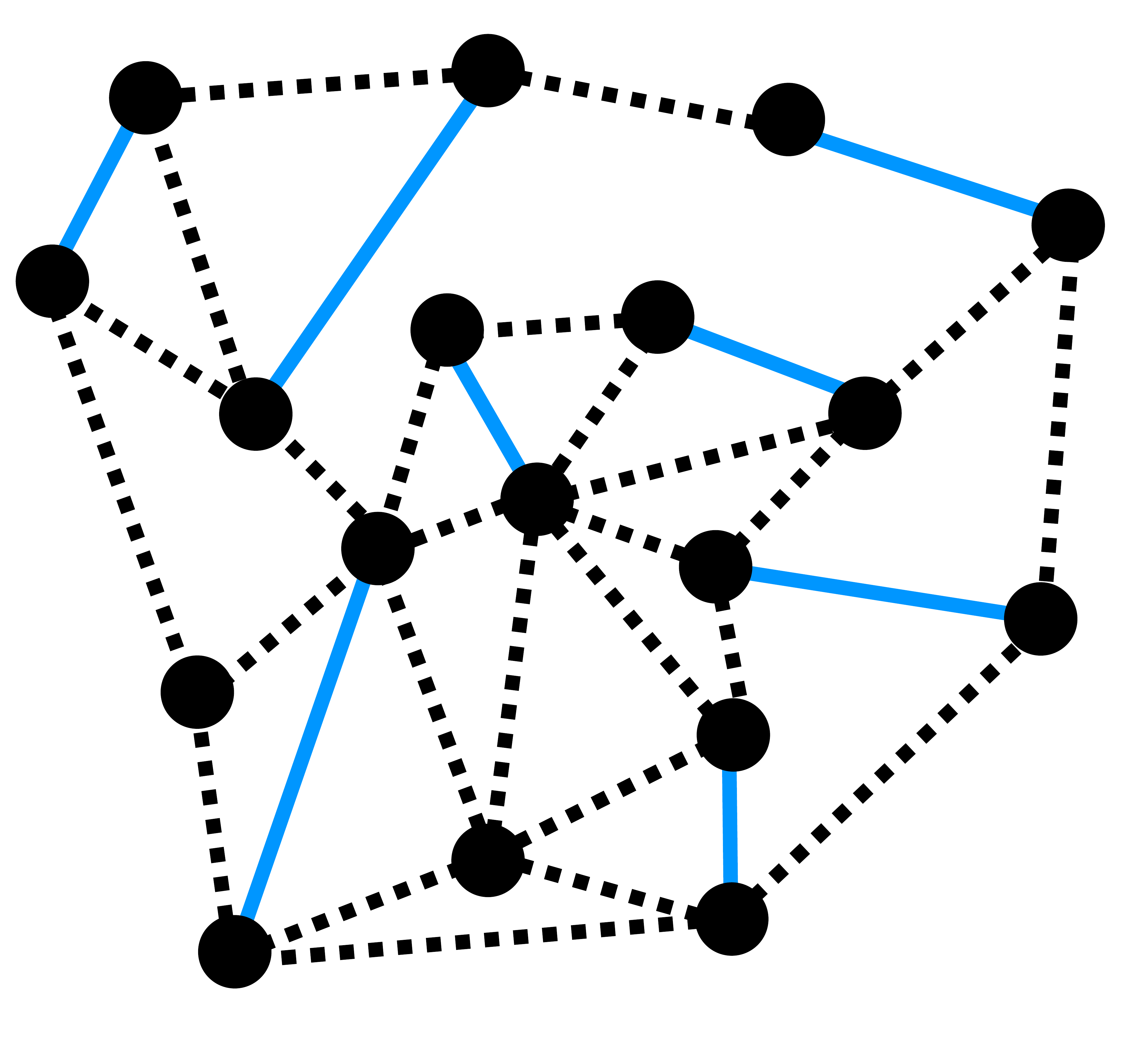}
        \caption{First matching.}\label{sfig:pg2}
    \end{subfigure}    \hfill
    \begin{subfigure}[b]{0.24\textwidth}
        \centering
        \includegraphics[width=\textwidth,trim=0mm 0mm 260mm 0mm, clip]{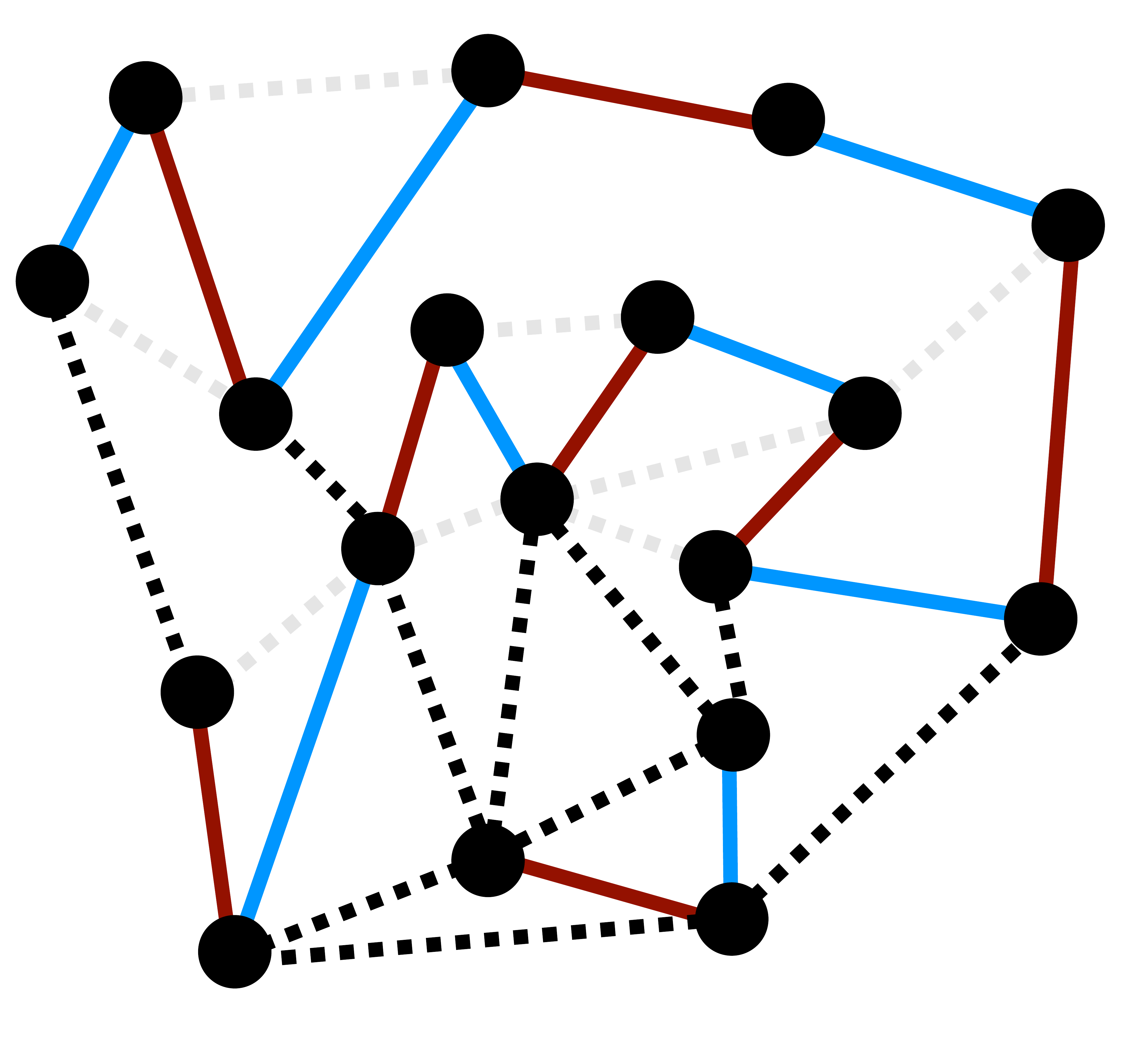}
        \caption{Second matching.}\label{sfig:pg3}
    \end{subfigure} \hfill
    \begin{subfigure}[b]{0.24\textwidth}
        \centering
        \includegraphics[width=\textwidth,trim=0mm 0mm 260mm 0mm, clip]{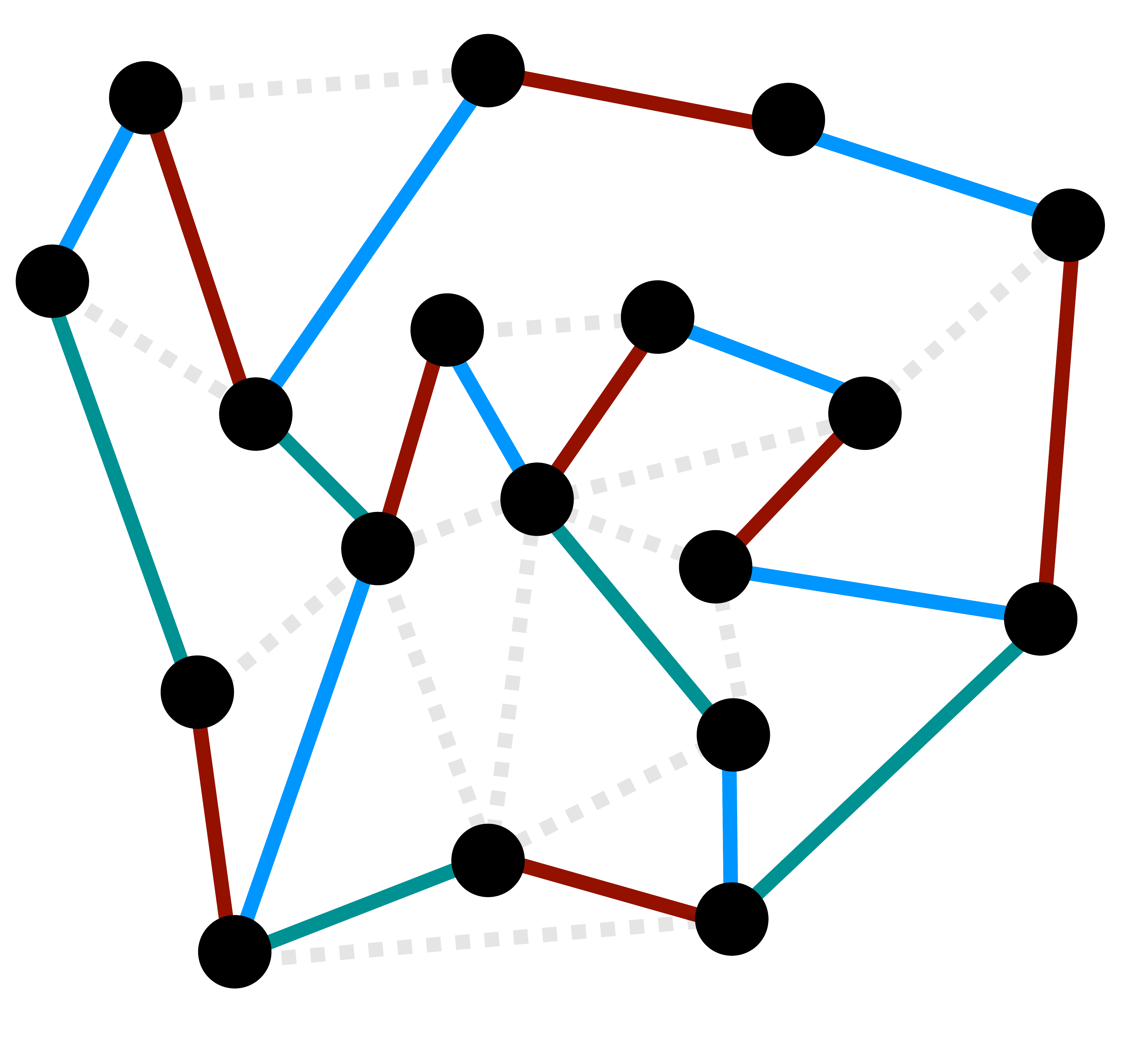}
        \caption{Third matching.}\label{sfig:pg4}
    \end{subfigure}
    \caption{Parallel greedy to construct a $3$-spanner. Edges not in a matching dashed and transparent if $3$-spanned with respect to $H$. Each matching colored and solid.} \label{fig:pg}
\end{figure}

Our main result is a proof of the sparsity of the output of the parallel greedy algorithm.
\begin{restatable}{theorem}{mainThm}\label{thm:main}
Parallel greedy (\Cref{alg:parGreedy}) outputs a $t$-spanner with $t^3\cdot \log^3 n\cdot n^{1 + O(1/t)}$ edges.
\end{restatable}
% \enote{O in exponent?}
% \enote{Change to $O(t)$ in exponent; get rid of $2t-1$}
\noindent Even stronger, we show that parallel greedy's output has arboricity at most $t^3\cdot \log^3 n\cdot n^{O(1/t)}$. For constant $t$, this recovers the usual $n^{1 + O(1/t)}$ bound. While we state our results for unweighted graphs, the above can be made to work for edge-weighted graphs by a standard bucketing trick of weights at a cost of $O(\log n)$ in the sparsity.

\paragraph*{Analysis Overview.} The above result seems somewhat surprising in light of the usual analysis of the (sequential) greedy algorithm. As mentioned above, the output of the (sequential) greedy algorithm has low sparsity on account of its large girth. However, a cycle on $4$ nodes demonstrates that the parallel greedy algorithm can output a spanner with girth as small as $4$, regardless of the value of $t$---see \Cref{fig:girth4}. As such, an entirely different approach to analyzing the sparsity of the output spanner is required.

\begin{figure}[h]
    \centering
    \begin{subfigure}[b]{0.2\textwidth}
        \centering
        \includegraphics[width=\textwidth,trim=0mm 0mm 315mm 0mm, clip]{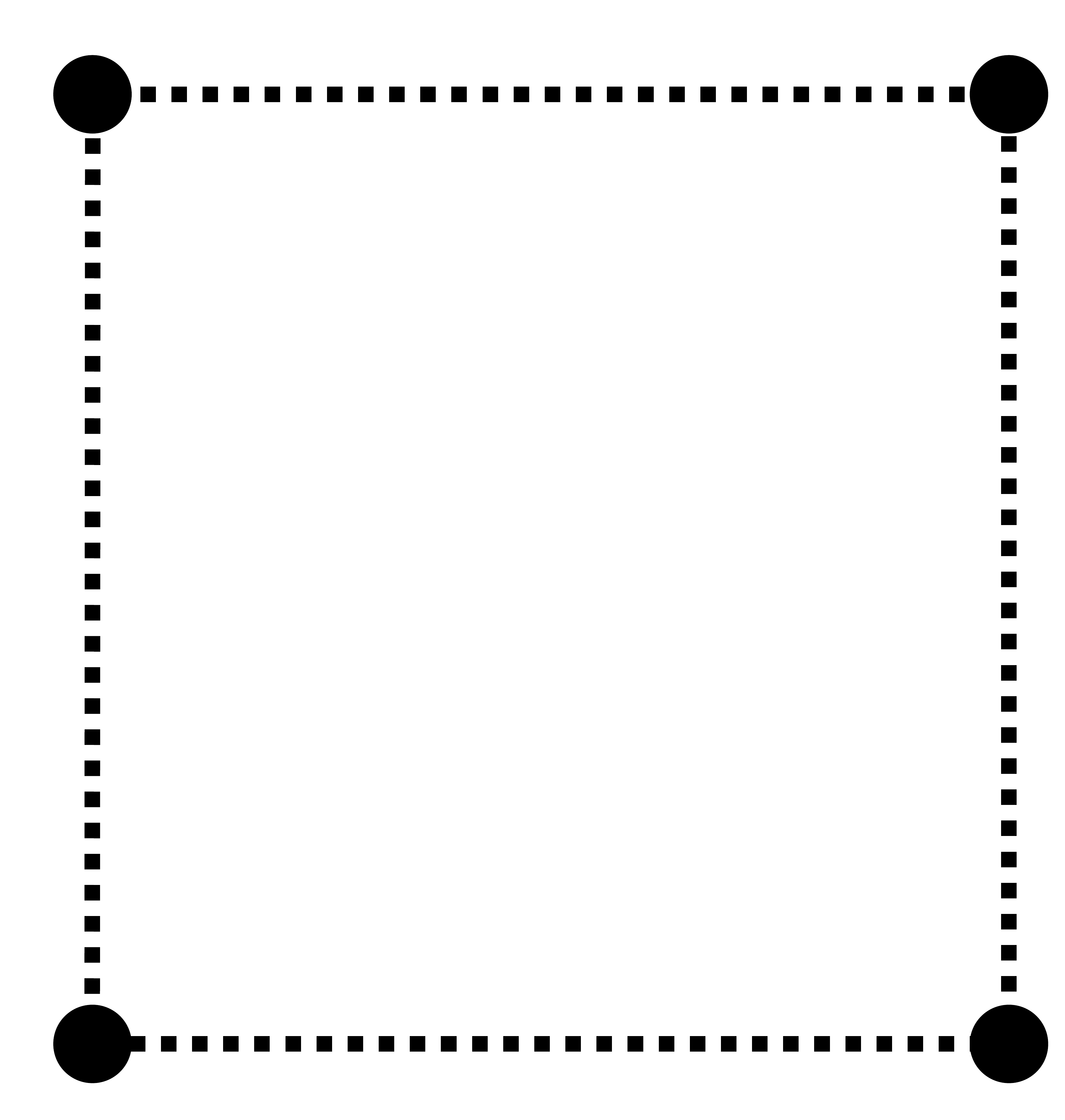}
        \caption{Input $G$.}\label{sfig:girthFour1}
    \end{subfigure}    \hfill
    \begin{subfigure}[b]{0.2\textwidth}
        \centering
        \includegraphics[width=\textwidth,trim=0mm 0mm 315mm 0mm, clip]{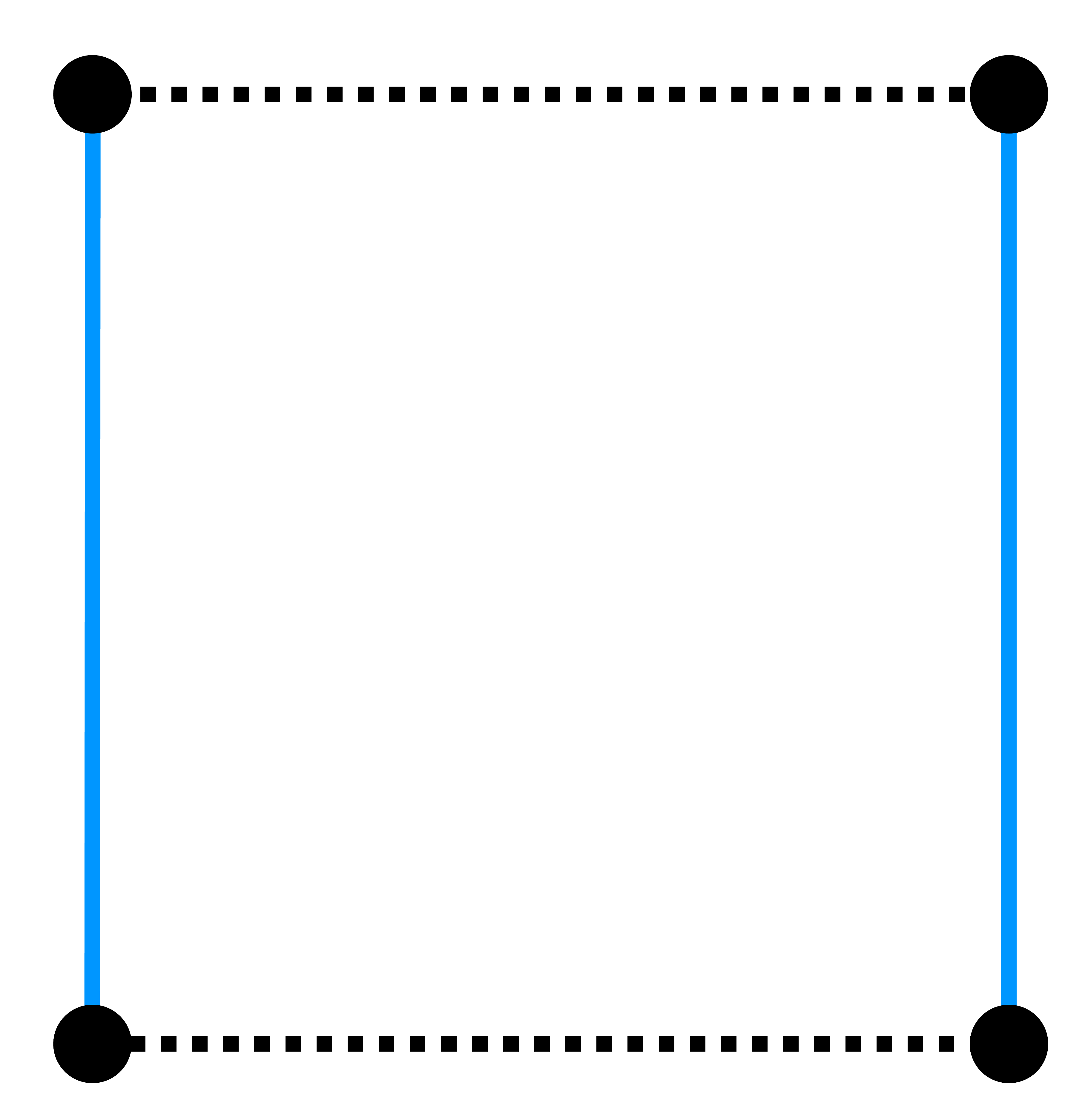}
        \caption{First matching.}\label{sfig:girthFour2}
    \end{subfigure}    \hfill
    \begin{subfigure}[b]{0.2\textwidth}
        \centering
        \includegraphics[width=\textwidth,trim=0mm 0mm 315mm 0mm, clip]{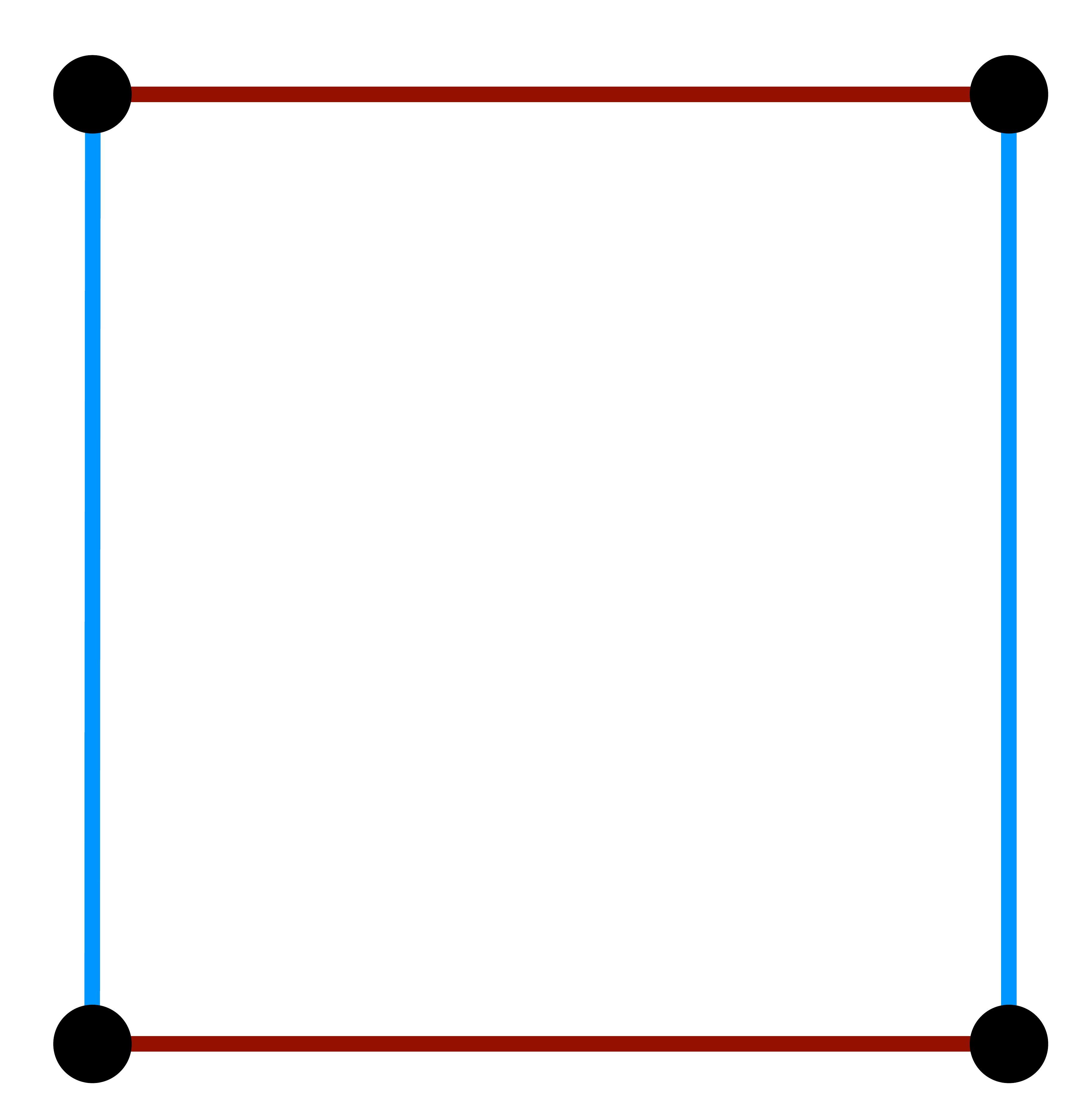}
        \caption{Second matching.}\label{sfig:girthFour3}
    \end{subfigure}\caption{A cycle on $4$ vertices shows that parallel greedy can return subgraphs with girth $4$ for any $t$. \ref{sfig:girthFour1} gives the input graph with $t$-unspanned edges dashed in black. \ref{sfig:girthFour2} and \ref{sfig:girthFour3} give the first and second matchings added by parallel greedy in blue and red respectively.} \label{fig:girth4}
\end{figure}

% By construction the spanner output by the greedy algorithm has girth at least $t$. The classic sparsity analysis \cite{althofer1993sparse,ahmed2020graph} then makes use of folklore (Moore) bounds which state that a graph with girth at least $t$ contains at most $n^{1 + O(1/t)}$ edges.
% However, one can use no such bounds for the parallel greedy algorithm since the output may have girth as small as $4$. 

 Instead of a girth-based argument, we make use of an analysis based on length-constrained expander decompositions, as recently introduced by \cite{haeupler2022hop}. Roughly our analysis is as follows. A length-constrained expander decomposition allows us to assume that (up to the deletion of a small number of edges), given a matching in a graph with large minimum degree, one can find many $t$-length paths between the matching endpoints such that no edge is used by too many paths. On the other hand, a graph with large arboricity has a subgraph with large minimum degree and so if the output of our algorithm has large arboricity then we can find said non-overlapping paths for the last matching chosen by our algorithm. Since these paths do not overlap too much, one of them must not use any edges in the matching itself, contradicting the $t$-unspannedness of the edges we choose. Thus, our analysis provides an alternative approach to analyzing the sparsity of spanners output by greedy algorithms that does not rely on delicate girth-based arguments.

\paragraph*{Applications and Additional Related Work.} While we make use of the above analysis to bound the sparsity of the parallel greedy algorithm, we believe that graphs produced by the parallel greedy algorithm---which we call $t$-$\pg$ graphs---are a fundamental extremal object which will find applications beyond this work. Indeed, while our approach makes use of (the existence of) length-constrained expander decompositions, it is also conversely useful for (algorithms for computing) length-constrained expander decompositions. Specifically, another work \cite{haeuplerBetterExpanders} makes use of the sparsity of $t$-$\pg$ graphs to argue that the union of a sequence of sparse (length-constrained) cuts is itself a sparse (length-constrained) cut. This structural fact, in turn, was used by \cite{haeuplerBetterExpanders} to give improved algorithms for computing length-constrained expander decompositions. 

Likewise, several works \cite{bodwin2018optimal,bukh2017bound,fernandez2020graph,bodwin2022partially,bodwin2022bridge} prove lower bounds by considering graphs that are produced by adding batches of edges (not necessarily a matching) where no edge can complete a short cycle. For instance, this approach has been used to prove lower bounds on fault tolerant spanners \cite{bodwin2018optimal} and the communication complexity of computing spanners \cite{fernandez2020graph}.

In summary, we show that the classic greedy algorithm admits significant opportunities for parallelism while retaining its sparsity guarantees. In the process of showing this, we introduce a new robust way of analyzing the density of a particular graph class which we expect to find applications beyond spanners. The remainder of this work is dedicated to showing \Cref{thm:main}.

\section{Notation and Conventions}\label{sec:notation}
We review the notation and conventions we make use of throughout this work.

\paragraph*{Graphs.}
Let $G=(V,E)$ be a graph and edge-length function $\l_G$.  If unspecified then $\l_G$ assigns every edge value $1$. We let $d_G$ denote the shortest path metric according to $\l_G$ in $G$. We let $n := |V|$ and $m:= |E|$ be the number of vertices and edges. The \emph{girth} of a graph is the minimum number of edges in a cycle. A \emph{matching} $M \subseteq E$ is a collection of pairwise disjoint edges. Given edges $C$ we let $G - C := (V, E \setminus C)$ and $G + C := (V, E\cup C)$. We say a subgraph is non-empty if it contains at least one edge. We let $E_G(U,W) := \{\{u,v\} \in E : u \in U, v \in W\}$ and let $\deg_G(v)$ give the degree of $v$ in $G$. For $U \subseteq V$, the induced subgraph $G[U]$ on $U$ is $(V, E(U))$ where $E(U)$ is all edges in $E$ with both endpoints in $U$. We drop $G$ subscript when it is clear from context.

\paragraph*{Arboricity.} A forest cover of graph $G$ is a collection of edge-disjoint subgraphs $F_1, F_2,\ldots, F_k $ of $G$ all of which are forests and whose union is $G$. $k$ is the size of the forest cover and $G$ is said to have arboricity $\alpha$ if its minimum size forest cover has size $\alpha$.

\paragraph{Flows.} A \emph{(multicommodity) flow} $F$ in $G$ is a function that assigns to each simple path $P$ in $G$ a flow value $F(P)\ge0$. We say $P$ is a \emph{flow-path} of $F$ if $F(P)>0$. 
% The value of $F$ is $\val(F) = \sum_{P} F(P)$.
%
The \emph{congestion of $F$ on an edge $e$} is defined to be $\congest_{F}(e)=\sum_{P:e\in P}F(P)$. The \emph{congestion} of $F$ is $\max_{e\in E(G)}\congest_{F}(e)$.
The \emph{length} (a.k.a.\ dilation) of $F$ is the maximum length of any of its flow-paths.
%The \emph{(maximum) step} of $F$ is $\step_{F}=\max_{P:F(P)>0}|P|$, which measures the maximum number of edges in all flow-paths of $F$.

\paragraph{Demands.} A \emph{demand} $D:V\times V\rightarrow\mathbb{R}_{\ge0}$ assigns
a non-negative value $D(v,w) \ge 0$ to each ordered pair of vertices in $V$. The load of demand $D$ is defined as $$\load(D) := \max_v \left(\max\left(\sum_{u} D(u,v), \sum_{u} D(v,u) \right)\right).$$ A demand $D$ is \emph{unit} if we have $\load(D) \leq \deg(v)$ for every $v\in V$. The size of a demand is written as $|D|$ and is defined as $\sum_{v,w} D(v,w)$. 
A demand $D$ is called \emph{$h$-length} constrained (or simply \emph{$h$-length}) in $G$ if it assigns positive values only to pairs that are within
distance at most $h$ in $G$. Given a flow $F$, the \emph{demand routed by $F$} is denoted
by $D_{F}$ where, for each $u,v\in V$, $D_{F}(u,v)=\sum_{P\text{ is a path from $u$ to $v$}}F(P)$
is the value of the flow from $u$ to $v$.

\paragraph{Routing Demands and Matchings with Flows.}
We say that a \emph{demand $D$ is routable in $G$ with congestion $\eta$ and dilation $h$} iff there exists a flow $F$ in $G$ with congestion $\eta$ and dilation $h$ where $D_{F}=D$. We say that a matching $M \subseteq E$ is $\delta$-routable with congestion $\eta$ and dilation $h$ if the demand $D_M(u,v) = \delta$ iff $\{u, v\} \in M$ and $0$ otherwise is routable with congestion $\eta$ and dilation $h$; here we imagine that each $e \in M$ has a canonical $u$ and $v$ so that $\load (D_M) \leq \delta$.

\section{High Degree Subgraphs with Routable Matchings}\label{sec: proof of linked ED}
Towards bounding the arboricity of the output of our parallel greedy algorithm, in this section we show that any graph of high minimum degree has a non-empty subgraph of essentially the same minimum degree where any matching can be routed with low dilation and congestion. Specifically, we show the following where $\congslack$ is a fixed universal constant hidden in \Cref{thm: linked ED}; later stated.

\begin{restatable}{lemma}{routSubgraph}\label{lem: linked ED}
There exist universal constants $\conConst$, $\congslack$ such that for every $t \geq 2$ if we let $$\fbox{$\phi :=1/(2 t\cdot n^{\congslack/ t} \cdot \log n)$} \qquad and \qquad \fbox{$\delta :=4\conConst\cdot t\log n/\phi^2 $}$$ then every $n$-vertex graph $H$ of minimum degree $\delta$, has a non-empty subgraph $H'\subseteq H$, such that 
\begin{enumerate}
%   \item $|E'|=\phi\cdot n^{O(1/g)}\cdot |E(H)|$;
    \item \textbf{Minimum Degree:} $H'$ has minimum degree at least \fbox{$\delta':= \frac{\phi}{2t}\cdot \delta$}; and
    \item \textbf{Routings:} Any matching $M \subseteq H'$ is $\delta'$-routable in $H'$ with dilation $t$ and congestion $\delta'/2$.%$O(\log n/\phi)$.% there exists a flow $F$ in $H'$ that sends, for each $(s,t)\in M$, $D'$ units of flow from $s$ to $t$ via paths of length $\le g$, that causes congestion at most $O(\log n/\phi)$. 
\end{enumerate}
\end{restatable}
Observe that $\delta'$-routing a matching with dilation $1$ and congestion $\delta'$ is trivial; notably the above improves this to congestion $o(\delta')$. The rest of this section proves \Cref{lem: linked ED} with length-constrained expander decompositions.
% $D = (4\congslack^2)\cdot g^3\cdot n^{(2\hopslack/ g)}\cdot \log^2 n$

\subsection{Preliminary: Length-Constrained Expander Decompositions}

Our proof of \Cref{lem: linked ED} makes use of length-constrained expander decompositions and so we begin by providing background on length-constrained expander decompositions, mostly from \cite{haeupler2022hop}. In short, a length-constrained expander decomposition is a small number of edge length increases so that every length-constrained unit demand in the resulting graph can be routed with low congestion and dilation. 

\subsubsection{Length-Constrained Cuts (a.k.a.\ Moving Cuts)}
We begin by giving formal definitions of length-constrained cuts which will allow us to define length-constrained expanders and length-constrained expander decompositions.

The following is the length-constrained analogue of a cut.
\begin{definition}[Length-Constrained Cut (a.k.a.\ Moving Cut)~\cite{haeupler2020network}]\label{def:movingcut}
An $h$-length cut (a.k.a. $h$-length moving cut) of graph $G = (V,E)$ is a function $C: E \mapsto \{0,\frac{1}{h},\frac{2}{h},\dots,1\}$. The \emph{size} of $C$ is $|C|:=\sum_{e} C(e)$. Any length-constrained cut with support in $\{0,1\}$ is called pure. 
\end{definition}

\noindent The following is the result of applying a length-constrained cut in a graph.
\begin{definition}[$G-C$]
For graph $G = (V,E)$ with edge length function $\l_G$ and length-constrained cut $C$, we let $G-C$ be $G$ with the edge-length function which assigns $e \in E$ value $\l_G(e) + h \cdot C(e)$. If $C$ is a pure cut then $G-C$ is $G$ with all edges in the support of $C$ deleted. %We refer to $G-C$ as $G$ \emph{after applying} $C$.
\end{definition}

\noindent The following gives the length-constrained analogue of separating a demand.
% \begin{definition}[$h$-Length Separation]
% Let $C$ be an $h$-length cut in graph $G= (V,E)$. We say nodes $u,v \in V$ are $h$-length separated by $C$ if $\dist_{G-C}(u,v)>h$.
% \end{definition}

\begin{definition}[$h$-Length Separated Demand]\label{dfn:sepDem}
For any demand $D$ and any $h$-length cut $C$, we define the amount of $h$-length separated demand as %the sum of demands between vertices that are $h$-length separated by $C$. We denote this quantity with $\sep_{h}(C,D)$,
$$\sep_{h}(C,D) = \sum_{u,v : \dist_{G-C}(u,v)>h} D(u,v).$$
\end{definition}

\noindent Using demand separation, we can define cut sparsity in the length-constrained setting.
\begin{definition}[$h$-Length Sparsity of a Cut $C$ for Demand $D$]\label{dfn:CDSparse}
For any demand $D$ and any $h$-length cut $C$ with $\sep_{h}(C,D)>0$, the $h$-length sparsity of $C$ with respect to $D$ is
$$\spa_{h}(C,D) = \frac{|C|}{\sep_{h}(C,D)}.$$
\end{definition}

Likewise, we can define the overall sparsity of a length-constrained cut as follows.
\begin{definition}[$(h,s)$-Length Sparsity of a Cut]\label{def:sparsity}
The $(h,s)$-length sparsity of $hs$-length cut $C$ is:
$$\spa_{(h,s)}(C) = \min_{\text{h-length unit demand}\ D} \spa_{s \cdot h} (C,D).$$
\end{definition}
% \noindent We refer to the minimizing demand $D$ above as the demand \emph{witnessing} the sparsity of $C$ with respect to $A$. We can analogously define the length-constrained conductance of a node-weighting.

% \noindent Note that if a demand $D$ has any demand between vertices that have length-distance exceeding $h$ then the empty cut has an $h$-length sparsity for $D$ which is equal to zero. For all other demands, i.e., for any non-empty $h$-length demand $D$ with $h < h'$, the $h'$-length sparsity of any cut $C$ for $D$ is always strictly positive.

\subsubsection{Length-Constrained Expanders}\label{sec:h-length-expander-def}

We now move on to formally defining length-constrained expanders. Informally, they are graphs with no sparse length-constrained cuts.

\begin{definition}[$(h,s)$-Length $\phi$-Expanders]
A graph $G$ is an $(h,s)$-length $\phi$-expander if every $hs$-length cut has $(h,s)$-length sparsity at least $\phi$.
\end{definition}

We now summarize the crucial properties of length-constrained expanders, namely the fact that they admit low congestion and dilation routings (see Lemma 3.16 of \cite{haeupler2022hop}).

\begin{theorem}
[Routing Characterization of Length-Constrained Expanders, \cite{haeupler2022hop}]\label{thm:flow character} Given graph $G$, for any $h \geq 1$, $\phi < 1$, and $s \geq 1$, there exists a universal constant $\conConst$:
\begin{itemize}
\item \textbf{Length-Constrained Expanders Have Good Routings:} If $G$ is an $(h,s)$-length $\phi$-expander, then every $h$-length unit demand can be routed in $G$ with congestion at most $\conConst \cdot \log n /\phi$ and dilation at most $s\cdot h$.
\item \textbf{Not Length-Constrained Expanders Have an Unroutable Demand:} If $G$ is not an $(h,s)$-length $\phi$-expander, then some $h$-length unit demand cannot be routed in $G$ with congestion at most $1/2\phi$ and dilation at most $\frac{s}{2}\cdot h$.
\end{itemize}
\end{theorem}

%B: Look under Routers: An expander is a $O(log n/phi)$-step $O(log n/phi)$-router
%\enote{I need a lemma about routing in normal expanders a la the below; I just picked a length slack that I assume is large enough:}
%\begin{lemma}[???]\label{lem:regExpEqualslengthExp}
%Suppose $G$ is a (normal) $\phi$-expander. Then $G$ is a $\left(\frac{\log N}{\phi}, 100 \right)$-length $\phi$-expander.
%\end{lemma}

\subsubsection{Length-Constrained Expander Decompositions}

Having defined length-constrained expanders, we can now define length-constrained expander decompositions. Informally, these are length-constrained cuts whose application renders the graph a length-constrained expander. 

More specifically, we will make use of a strengthened version of length-constrained expander decompositions called ``linked'' length-constrained expander decompositions. Informally, this is a length-constrained expander decomposition which renders $G$ length-constrained expanding even after adding many self-loops. This is a strengthened version because adding self-loops only makes it harder for a graph to be a length-constrained expander. 

% \begin{definition}[Length-Constrained Expander Decomposition] 
% Given graph $G$, an \emph{$(h,s)$-length $\phi$-expander decomposition} with cut slack $\kappa$ and length slack $s$ is an $hs$-length cut $C$ of size at most $\kappa \cdot \phi m$ such that $G-C$ is an $(h,s)$-length $\phi$-expander.
% \end{definition}

The following definition gives the self-loops we will add for a length-constrained expander decomposition $C$.
\begin{definition}[Self-Loop Set $L^{\ell}_C$]
Let $C$ be an $h$-length cut of a graph $G = (V,E)$ and let $\ell$ be a positive integer divisible by $h$. For any vertex $v$, define $C(v)=\sum_{e \ni v}C(e)$. The self-loop set $L^{\ell}_C$ consists of $C(v)\cdot \ell$ self-loops at $v$. We let $G+L^{\ell}_C := (V, E \cup L^{\ell}_C)$.
\end{definition}
\noindent Using the above self-loops, we can now define linked length-constrained expander decompositions.

\begin{definition}[Linked Length-Constrained Expander Decomposition]
Let $G$ be a graph. An $\ell$-linked $(h,s)$-length $\phi$-expander decomposition with cut slack $\kappa$ is an $hs$-length cut $C$ such that $|C|\le \kappa \cdot \phi m$ and $G + L^{\ell}_C - C$ is an $(h,s)$-length $\phi$-expander.
\end{definition}
Prior work of \cite{haeupler2022hop} demonstrated that for each $\phi$, a length-constrained expander decomposition always exist with length slack $s = \Omega(n)$, cut slack $\kappa = O(\log n)$ and linkedness $\ell \geq \Omega(1/\phi\cdot\log n)$.

\subsection{Existence of Length-Constrained Expander Decompositions}
In this work we will use the existence of pure length-constrained expanders which trade-off between $s$ and $\kappa$. Namely, we show the following theorem; our proof essentially follows that of \cite{haeupler2022hop} but accounts for the pureness of our decompositions (paying an extra $h$ in $\kappa$) and uses a more general form of the ``exponential demand''.
\begin{theorem}[Length-Constrained Expander Decompositions]
% \cite{haeupler2022hop,haeupler2022cut}
\label{thm: linked ED}
    There exists a constant $\congslack$ such that given graph $G = (V,E)$ with edge lengths, $h \geq 1$, $s \geq 100$, $\phi \geq 0$ and any $\ell \leq 1/(100 \phi \log n)$, there exists a pure $\ell$-linked $(h,s)$-length $\phi$-expander decomposition with cut slack $hs \cdot n^{\congslack/s} \cdot \log n$.
    % a set of edges $C \subseteq E$ such that $G - C := (V , E \setminus C, w)$ is an $(h,s)$-length $\phi$-expander where $|C| \leq hs \cdot n^{O(1/s)} \cdot \log n \cdot \phi m$.
\end{theorem}
\begin{proof}
%Let $l$ be the length function on edges of $G$. 
We denote by $\dist(\cdot,\cdot)$ the shortest-path distance between edges in $G$. That is, for a pair $e,e'\in E(G)$, $\dist(e,e')$ is the smallest total length of any path starting with $e$ and ending with $e'$. %\enote{Naming issue with $l$ vs $\ell$-linked? I think we use $l_G$ above for lengths}

We now define a demand $D$ on $G$ as follows. 
%\enote{This is all edge-based but our paper is all node-based} 
For a pair $e,e'$ of edges in $G$, we set $w(e,e')=n^{-\dist(e,e')/(sh/2)}$ if $\dist(e,e')\le sh/2$, otherwise we set $w(e,e')=0$. Denote $w(e)=\sum_{e'}w(e,e')$, and define $D(e,e')=w(e,e')/w(e)$, so for each $e$, $\sum_{e'}D(e,e')=1$. 
Intuitively, we can think of $D$ as a demand that sends, for each pair $e,e'$, $D(e,e')$ units of flow from $e$ to $e'$. Formally, we define the demand on vertices as follows. For every pair $v,v'$ of vertices in $G$, we send $D(v,v')=\sum_{e\sim v,e'\sim v'}D(e,e')$ units of flow from $v$ to $v'$. As $\sum_{e'}D(e,e')=1$ holds for all $e$, we get that the demand $D$ we defined on pairs of vertices is unit.
We prove the following two claims.

\begin{claim}
For every pair $e,e'$ with $\dist(e,e')\le h$, $\sum_{e''}\min\set{D(e,e''),D(e',e'')}\ge n^{-4/s}$.
\end{claim}
\begin{proof}
Observe that, for every $e''$, 
\[w(e',e'')=n^{-\frac{\dist(e',e'')}{(sh/2)}}\ge n^{-\frac{\dist(e,e')+\dist(e,e'')}{(sh/2)}}\ge n^{-\frac{h}{(sh/2)}}\cdot n^{-\frac{\dist(e,e'')}{(sh/2)}}= n^{-2/s}\cdot w(e,e''),\]
and as a corollary,
\[w(e')=\sum_{e''}w(e',e'')=\sum_{e''}n^{-\frac{\dist(e',e'')}{(sh/2)}}\ge \sum_{e''}n^{-2/s}\cdot w(e,e'')=n^{-2/s}\cdot w(e).\]
Therefore,
\[\begin{split}
\sum_{e''}\min\set{D(e,e''),D(e',e'')} & \ge  \sum_{e''}\min\set{\frac{w(e,e'')}{w(e)},\frac{w(e',e'')}{w(e')}} \\
& \ge  \sum_{e''}\frac{n^{-2/s}\cdot w(e,e'')}{n^{2/s}\cdot w(e)} \\
& \ge  n^{-4/s}.
\end{split}
\]
\end{proof}

%\enote{The $\spa$ notation here has no parens around $h$ and $s$ in subscript but we do have parens in the rest of the paper}
\begin{claim}
\label{clm: sparse wrt D}
Let $C$ be any moving cut with $\spa_{(h,s)}(C)\le \phi$. Then $\spa_{h,s/2}(C,D)\le 2\phi n^{4/s}$.
\end{claim}
\begin{proof}
Let $D^*$ be a unit $h$-hop demand with $\spa_{(h,s)}(C,D^*)\le \phi$. As $D^*$ is a unit demand, we first distributed the demand on each vertex to its incident edges, such that each edge sends and receives at most $1$ unit of flow. Consider the resulting demand $D^{**}$ on edges.
%\enote{What does it formally mean to simulate a demand?}
Note that, by concatenating a demand of $D^{**}(e,e')\cdot \min\set{D(e,e''),D(e',e'')}$ from $e$ to every $e''$ and a demand from $e''$ to $e'$ with the same amount, where we essentially applied the demand $D$ twice (where the second time we use it in the opposite direction), we obtain a demand that sends at least $D^{**}(e,e')\cdot n^{-4/s}$ units of flow from $e$to $e'$, for every pair $e,e'$.
If $e$ and $e'$ are at
distance at least $sh$ in $G-C$ then either the pair $e,e''$ or the pair $e'', e$ must be at distance
at least $sh/2$. Therefore,
$\sep_{hs/2}(C,D)\ge (n^{-4/s}/2)\cdot \sep_{hs}(C,D^*)$, and so $\spa_{(h,s/2)}(C,D)\le 2n^{-4/s}\cdot \spa_{(h,s)}(C,D^*)\le 2\phi n^{4/s}$.
\end{proof}

We now provide the proof of \Cref{thm: linked ED} by a simple algorithm. While there exists a moving cut $C$ in $G$ that is $(h,s)$-hop $\phi$-sparse, we let $\bar C$ be the pure cut corresponding to $C$ (that is, $\bar C$ contains all edges with non-zero values in $C$, so $|\bar C|\le hs\cdot |C|$), update $G\leftarrow G-\bar C+L^{\ell}_{\bar C}$ and continue. It is easy to verify that our algorithm ends with $G$ being an $\ell$-linked $(h,s)$-hop $\phi$-expander. 

It suffices to count the total number of deleted edges, for which we define a potential function as follows. For each edge (including self-loop) in $G$, we define its potential to be $\ln(w(e))\cdot 4hsn^{4/s}$. For every unit of moving cut that we applied to $G$, we define its potential to be $1/\phi$. Initially, the total potential is at most $|E(G)|\cdot \log |E(G)|\cdot 4hs\cdot n^{4/s}$ (each edge has potential at most $\log |E(G)|\cdot 4hs\cdot n^{4/s})$ and there is no moving cut applied. We will show that, after each iteration (where a cut is applied to $G$ and corresponding self-loops are added), the potential never increases. Note that this implies that the total number of deleted edges is at most $\phi\cdot|E(G)|\cdot \log |E(G)|\cdot 4hs\cdot n^{4/s}$, completing the proof of \Cref{thm: linked ED} (as $|E(G)|\le n^2$).

According to the algorithm, in an iteration where a sparse cut $C$ is found, we apply $\bar C$ to $G$ and add self-loops, which increases the potential by at most
\[|C|\cdot hs/\phi + |C|\cdot hs\cdot \ell\log n\le 2\cdot |C|\cdot hs/\phi,\]
as $\ell \leq 1/(100 \phi \log n)$.
On the other hand, the potential decrease is at least
\[
4n^{4/s}\cdot hs\cdot \sum_{e\in E(G)}\ln (w(e)) - \ln (w'(e))\ge 4n^{4/s}\cdot hs\cdot \sum_{e\in E(G)}\frac{w(e) - w'(e)}{w(e)},
\]
where $w'(e)$ represents the weight of $e$ in the updated graph $G$.
From \Cref{clm: sparse wrt D}, we know that the $(h,s/2)$-hop sparsity of $C$ for $D$ is at most $2\phi n^{4/s}$. This means that $\sep_{sh/2}(C,D)\ge |C|/(2\phi n^{4/s})$.
On the other hand, by definition of $D$, the right-hand-side of the above inequality is at least the amount of demand that is $sh/2$-separated by $C$ (as those demand will have weight $w'(e,e')=0$ according to our definition), and so is at least $|C|/(2\phi n^{4/s})$. Altogether, we get that the potential decrease is at least 
\[
4n^{4/s}\cdot hs\cdot |C|/(2\phi n^{4/s})\ge 2\cdot |C|\cdot hs/\phi,
\]
which is an upper bound of the potential increase. This completes the proof that the potential never increases.
\end{proof}

\subsection{High Degree Routable Matching Subgraphs (\Cref{lem: linked ED} Proof)}
We conclude this section with our proof of \Cref{lem: linked ED}.

\routSubgraph*
\begin{proof}
The basic idea of the proof is as follows. We first take a linked length-constrained expander decomposition. Then, the result of applying this length-constrained expander decomposition must have high minimum degree because if any vertex $v$ has its degree drop too low we can find a sparse length-constrained cut (namely the singleton cut separating that $v$). Likewise, matchings are routable with low dilation and congestion because the result of applying our length-constrained expander decomposition is a length-constrained expander.

We begin by more formally describing our length-constrained expander decomposition. Specifically, we apply \Cref{thm: linked ED} to graph $H$ with parameters $h=1$, $s=t$ and $\ell=1/(100 \phi\log n)$. We let $C$ be the resulting pure cut which is a $1/(100\phi\log n)$-linked $(h,s)$-length $\phi$-expander decomposition for $H$ with cut slack $t \cdot n^{\congslack/t} \cdot \log n$.

We first define $H'$. Observe that by our choice of $\phi$ we have
\begin{align*}
    |C| & \leq t \cdot n^{\congslack/t} \cdot \log n \cdot \phi \cdot |E(H)|\\
    &\leq \frac{|E(H)|}{2}
\end{align*}
and therefore $H- C$ contains at least one edge. We let $H'$ be an arbitrary connected component of $H-C$ with at least one edge. It remains to show that $H'$ satisfies the properties required in \Cref{lem: linked ED}.

% Let $E'$ be the set of all edges $e\in E(H)$ with $C(e)\ne 0$, so 
% \begin{align*}
%     |E'|&\le t\cdot |C|\\
%     &\le t\cdot n^{\congslack/ t}\cdot \phi \cdot |E(H)|
% \end{align*}

% \enote{I'm not sure how $E'$ is different than $C$ if we're assuming that $C$ is already pure? (mod function on edges vs edge set). The inequality makes it sound like $C$ is not pure}. \znote{You are right, $E'=C$ as $C$ is pure. We might want to keep the notation $C$ only because later on we will use $L^{\ell}_C$ to denote self-loops}
% As \fbox{$\phi=1/(2\congslack\cdot g\cdot n^{(\hopslack/ g)})$}, $|E'|\le |E(H)|/2$ and therefore $H\setminus E'$ is non-empty.
% We let $H'$ be any connected component of $H\setminus E'$ that contains more than one vertex. 

 %\enote{How does pureness of cut interact with self-loops added? Don't think can add self-loops after scaling. It looks like we're adding them before making the cut pure, is that right?}
%\znote{Yes. It is like: recursively compute a sparse moving cut, turn it into a pure cut and apply to the graph, and then add self-loops (according to the pure cut).}

We first show that the minimum degree in $H'$ is at least $\delta'$. Assume for the sake of contradiction that there is a vertex $v$ in $H'$ such that $\deg_{H'}(v) < \delta'$. Since $H'$ is a connected graph with at least two vertices, $v$ has at least one neighbor in $H'$; let $u$ be an arbitrary such neighbor. By our choice of $\phi$ and $l$, we have that 
\begin{align}\label{eq:aa}
 l = 1/(100\phi \log n) = t \cdot n^{\congslack/t}/50 = \omega(1) \geq 2   
\end{align}
and
\begin{align}\label{eq:ab}
    \delta' = \frac{\phi}{2t} \cdot \delta = \frac{1}{4t^2 \cdot n^{\congslack/t}\cdot \log n} \cdot \delta\leq \frac{\delta}{2}.
\end{align}
Thus, from the definition of the set $L_C^{\ell}$ and our choice of $\delta'$ and $\phi$ and applying Equations \ref{eq:aa} and \ref{eq:ab}, the number of self loops at $v$ in $H-C+L_C^\ell$ is at least
\begin{align*}
    (\delta-\delta') \cdot l \geq \delta.
\end{align*}
Likewise, the degree of $u$ in $H-C+L_C^\ell$ is at least $\delta$ since in $H$ it has degree at least $\delta$, and if the cut $C$ assigns non-zero values to $d$ of its incident edges, then $d\ell>d$ self-loops will be added to $u$. 

Next, consider the demand $D_0$ which gives value $\delta$ to $(u,v)$ and $0$ to all other pairs of vertices. Since both $u$ and $v$ have degree at least $\delta$ in $H-C+L_C^\ell$, it follows that $D_0$ is unit in $H-C+L_C^\ell$. Likewise, since $u$ and $v$ are adjacent in $H'$, $D_0$ is a $1$-length demand in $H-C+L_C^\ell$.

Let $C_0$ be the length-constrained cut that assigns value $t$ to all edges in $E_{H'}(\set{v},V(H')\setminus \set{v})$ and value $0$ to all other edges. Since $v$ has degree at most $\delta'$ in $H'$ we have 
\begin{align}\label{eq:ba}
    |C_0| \leq t \cdot \delta'.
\end{align}

However, $C_0$ $t$-separates $D_0$ in $H-C+L_C^\ell$ and since $D_0$ is a unit $1$-length demand and $H-C+L_C^\ell$ is a $(1,t)$-length $\phi$-expander it follows that
\begin{align}\label{eq:bb}
    |C_0| \geq \phi \cdot |D_0| = \phi \cdot \delta.
\end{align}
Combining Equations \ref{eq:ba} and \ref{eq:bb} we get
\begin{align*}
    \phi \cdot \delta \leq t \cdot \delta '
\end{align*}
which contradicts the definition of $\delta'$ as $\frac{\phi}{2t} \cdot \delta$.

It remains to show that any matching in $H'$ is $\delta'$-routable with dilation $t$ and congestion at most $\delta'/2$. Consider matching $M \subseteq H'$ and let $D_M$ be the demand which assigns $(u,v)$ value $\delta'$ for each $\{u,v\} \in M$ (for a canonical ordering of the vertices). $D_M$ is $1$-length by construction and unit in $H'$ since every vertex in $H'$ has degree at least $\delta'$. Furthermore, $H'$ is a $(1,t)$-length $\phi$-expander since $H-C+L_C^\ell$ is and so $D_M$ can be routed with dilation $t$ and congestion  
\begin{align*}
\conConst\cdot\log n/\phi\le  \delta'/2
\end{align*}
in $H'$.
\end{proof}

\section{Bounding the Arboricity of Parallel Greedy Graphs}
Having shown how to route in high minimum degree graphs in the previous section, in this section we bound the arboricity of any graph constructed in the manner of the parallel greedy algorithm. We abstract out such graphs with the notion of a parallel greedy graph.

\begin{definition}[$t$-$\pg$ Graphs]
    Let $V$ be a set of vertices. We say that a sequence  of edge sets $(E_1,\ldots,E_k)$ on $V$ is $t$-$\pg$ for some integer $t\ge 2$, iff for each $i \in [k]$,
\begin{itemize}
    \item $E_i$ is a matching on $V$; and
    \item $\dist_{G_{i-1}}(u,v)> t$ for each $(u,v)\in E_i$ where $G_{i-1}$ the graph on $V$ with edges $\bigcup_{1\le j\le i-1}E_j$.
\end{itemize}
We say that a graph $G = (V,E)$ is $t$-$\pg$ iff its edge set $E$ is the union of some $t$-$\pg$ sequence on $V$.
\end{definition}
% \noindent Equivalently, a sequence $(E_1,\ldots,E_k)$ is $t$-$\pg$ iff for each $i \in [k]$, every cycle in $G_i$ of length at most $t+1$ contains at least two edges in $E_i$.

The following summarizes our bound on the arboricity of $t$-\pg graphs.
\begin{restatable}{theorem}{gPGArb}\label{thm: main}
    Every $t$-\pg graph on $n$ vertices has arboricity $t^3\cdot \log^3n \cdot n^{O(1/t)}$.
\end{restatable}

% \begin{theorem}
% \label{thm: main}
% Every $t$-\pg graph on $n$ vertices has arboricity $n^{O(1/t)}$.
% \end{theorem}
The remainder of this section is dedicated to a proof of \Cref{thm: main}. In particular, we will observe that if our parallel greedy graph has high arboricity then it has a high minimum degree subgraph and this subgraph admits low congestion and dilation routings (by \Cref{lem: linked ED}). We will then use these routings to contradict the parallel greediness of the input graph.

\subsection{Preliminary: Minimum-Degree Subgraphs from Arboricity}
We begin by noting a known fact about high minimum degree subgraphs of high arboricity graphs.
\begin{restatable}{lemma}{highMinDeg}\label{lem:highMinDeg}
    Let $G = (V,E)$ be a graph with arboricity $\alpha$. Then, $G$ has a non-empty induced subgraph with minimum degree at least $\alpha/2$.
\end{restatable}
\noindent The rest of this section is dedicated to showing this folklore fact using some standard arguments.

To show \Cref{lem:highMinDeg}, we will make use of a famous result of Nash-Williams characterizing graph arboricity in terms of graph density as well as two simple helper lemmas.
\begin{theorem}[\cite{nash1961edge,nash1964decomposition,chen1994short}]\label{lem:NW}
    Graph $G = (V,E)$ has arboricity at most $\alpha$ iff for every $U \subseteq V$ we have
    \begin{align*}
        |E(U)| \leq \alpha \cdot (|U|-1).
    \end{align*}
\end{theorem}
Likewise we make use of the following two simple helper lemmas.
\begin{lemma}\label{lem:arbHelp}
    Let $G = (V,E)$ be a connected graph with arboricity $\alpha$. Then $|E| + 1 \geq |V| -1 + \alpha$.
\end{lemma}
\begin{proof}
    Let $F_1, F_2, \ldots, F_\alpha$ be the minimum size forest decomposition of $G$. We may assume that $F_1$ is a spanning tree of $G$ since if it is not we can always move edges from $F_i$ for $i \geq 2$ to $F_1$ to make it a spanning tree. Furthermore, we may assume that each $F_i$ contains at least one edge (otherwise we would have violated the definition of arboricity). It follows that $F_1$ contains $|V|-1$ unique edges and $F_i$ for $i \geq 2$ contains at least one unique edge, giving the inequality.
\end{proof}
\begin{lemma}\label{lem:minDeg}
    Let $G = (V,E)$ be a graph with average degree $\rho := \frac{\sum_{u} \deg_G(u)}{|V|}$. Then there is a non-empty $U \subseteq V$ such that $G[U]$ has minimum degree at least $\rho/2$.
\end{lemma}
\begin{proof}
    Our proof is by the following construction: initialize $U$ to $V$; while there exists a vertex $u \in U$ such that $\deg_{G[U]}(u) < \rho / 2$ remove $u$ from $U$. By construction the minimum degree of $G[U]$ is at least $\rho/2$. 
    
    To see that $U$ is non-empty, let
    \begin{align*}
        \rho_U := \frac{\sum_{u \in U} \deg_{G[U]}(u)}{|U|}
    \end{align*}
    be the average degree of $U$. It suffices to show that our construction always satisfies $\rho_U > 0$ since if $U$ is empty at the end of our construction then at some point the numerator of $\rho_U$ would be $0$ and the denominator would be $1$ (namely, when $|U| = 1$). This follows since initially $\rho_U$ has value $\rho$ and each time we remove a vertex from $u$ we decrease the numerator of $\rho_U$ by strictly less than $\rho$ and the denominator by $1$.
\end{proof}

Using the above theorem of Nash-Williams and our two helper lemmas, we conclude with a proof of \Cref{lem:highMinDeg}.
\highMinDeg*
\begin{proof}
    By \Cref{lem:NW} there is a non-empty $U \subseteq V$ such that
    \begin{align}\label{eq:cb}
        |E(U)| \geq (\alpha-1) \cdot (|U| - 1) + 1
    \end{align}
    since other we could decrease $\alpha$ by $1$ and still satisfy the inequality in \Cref{lem:NW}, contradicting \Cref{lem:NW}. We can assume without loss of generality that $U$ is connected. Applying \Cref{lem:arbHelp} and the fact that $G[U]$ has arboricity at least $\alpha$ we have
    \begin{align}\label{eq:cc}
        |E(U)| + 1 \geq |U| - 1 + \alpha
    \end{align}
    
    Thus, combining Equations \ref{eq:cb} and \ref{eq:cc} we have
    \begin{align*}
        2|E(U)| \geq \alpha \cdot|U|
    \end{align*}
    
    It follows that the average degree of a vertex in $G[U]$ is 
    \begin{align*}
        \frac{\sum_{u \in U} \deg_{G[U]}(u)}{|U|} = \frac{2 |E(U)|}{|U|} \geq \alpha.
    \end{align*}
    and so by \Cref{lem:minDeg} we know that $G[U]$ and therefore $G$ contains a subgraph of minimum degree at least $\alpha/2$.
\end{proof}

\subsection{Arboricity of Parallel Greedy Graphs (\Cref{thm: main} Proof)}

In this section, we provide the proof of \Cref{thm: main}. We make use of \Cref{lem: linked ED} and \Cref{lem:highMinDeg} from previous sections. We also make use of the following simple helper lemma.

\begin{lemma}
\label{lem: pg subgraph}
Every subgraph of a $t$-\pg graph is also a $t$-\pg graph.
\end{lemma}
\begin{proof}
Let $G$ be a $t$-\pg graph and let $H$ be a subgraph of $G$.
Let $(E_1,\ldots,E_k)$ be the $t$-\pg sequence that generates $G$. For each $i \in [k]$, let $E'_i$ be the subset of $E_i$ that contains all edges lying in $H$. Clearly, $E'_i$ is a matching in $H$, and if we denote by $H_{i-1}$ the graph induced by edges in $\bigcup_{1\le j\le i-1}E'_j$, then for every edge $(u,v)\in E'_i$ we have
\begin{align*}
    \dist_{H_{i-1}}(u,v)> t.
\end{align*}
This follows because otherwise we have 
\begin{align*}
\dist_{G_{i-1}}(u,v)\le \dist_{H_{i-1}}(u,v)\le t
\end{align*}
and since $H_{i-1}$ is a subgraph of $G_{i-1}$ and $(u,v)$ is also an edge in $E_i$, this leads to a contradiction to the assumption that $(E_1,\ldots,E_k)$ is a $t$-\pg sequence.
\end{proof}

We conclude with our proof of \Cref{thm: main}, illustrated in \Cref{fig:pgProof}.

\begin{figure}[h]
    \centering
    \begin{subfigure}[b]{0.24\textwidth}
        \centering
        \includegraphics[width=\textwidth,trim=0mm 0mm 260mm 0mm, clip]{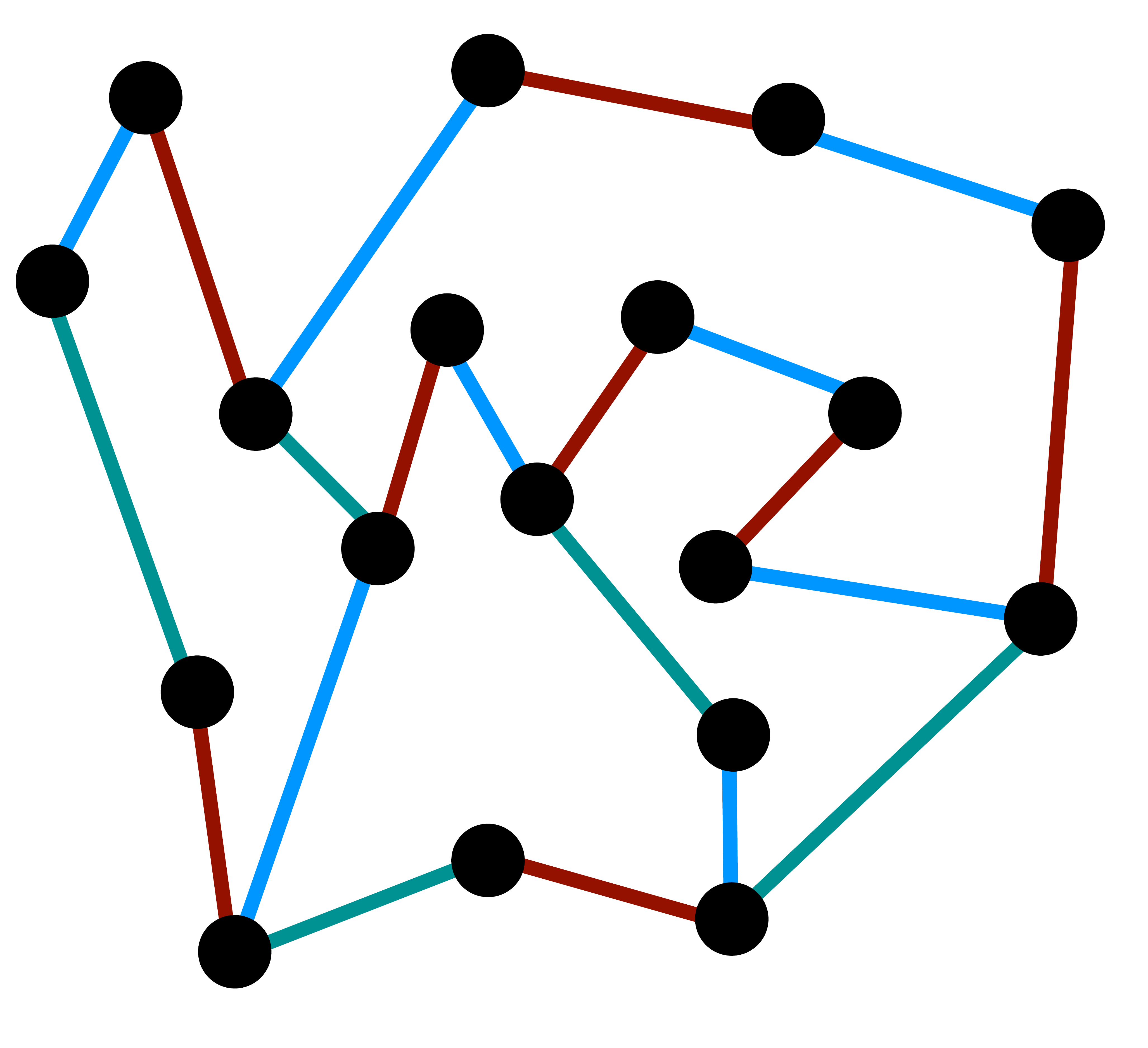}
        \caption{Input $t$-\pg graph $G$.}\label{sfig:pgProof1}
    \end{subfigure}    \hfill
    \begin{subfigure}[b]{0.24\textwidth}
        \centering
        \includegraphics[width=\textwidth,trim=0mm 0mm 260mm 0mm, clip]{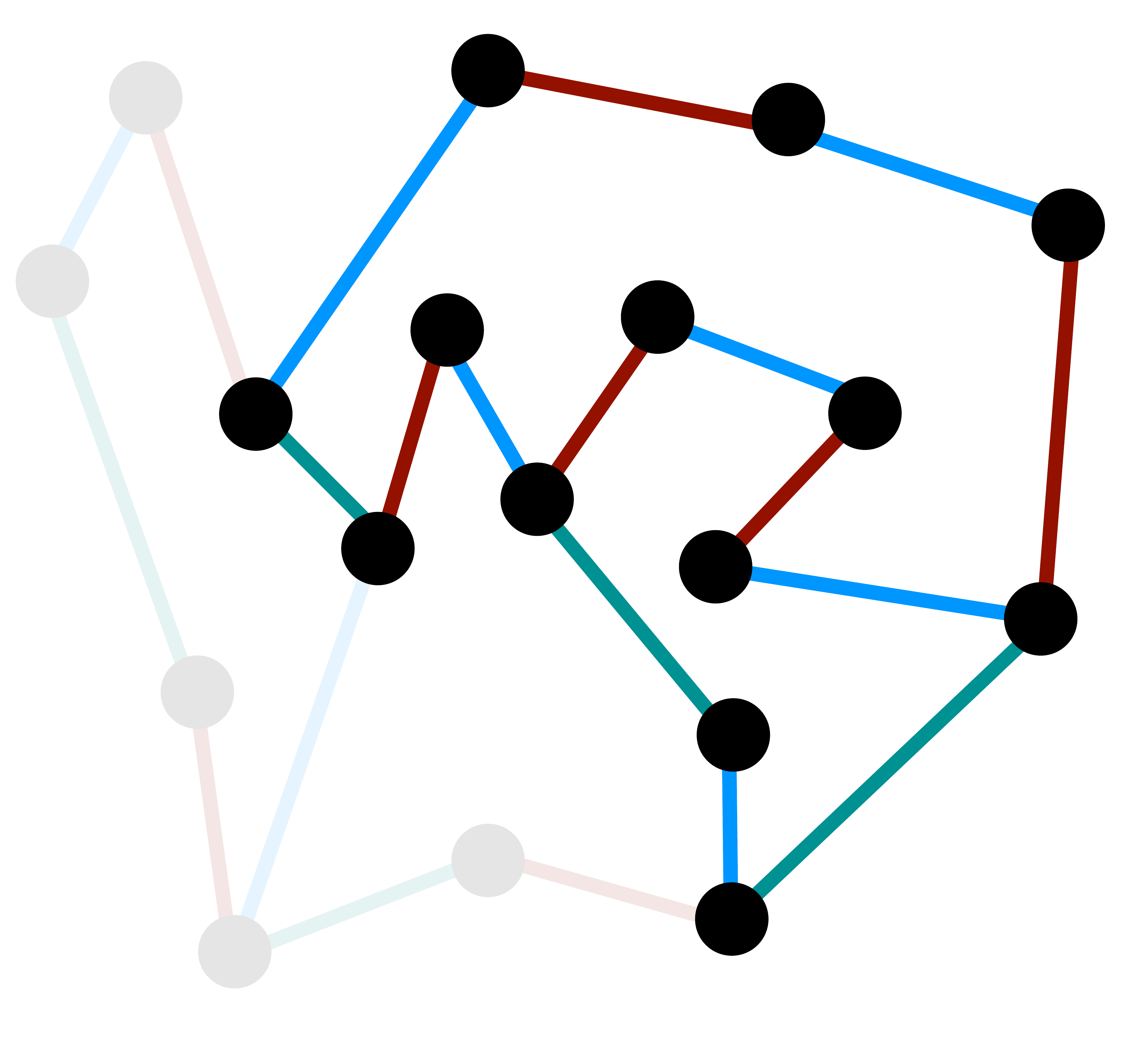}
        \caption{Subgraph $H \subseteq G$.}\label{sfig:pgProof2}
    \end{subfigure}    \hfill
    \begin{subfigure}[b]{0.24\textwidth}
        \centering
        \includegraphics[width=\textwidth,trim=0mm 0mm 260mm 0mm, clip]{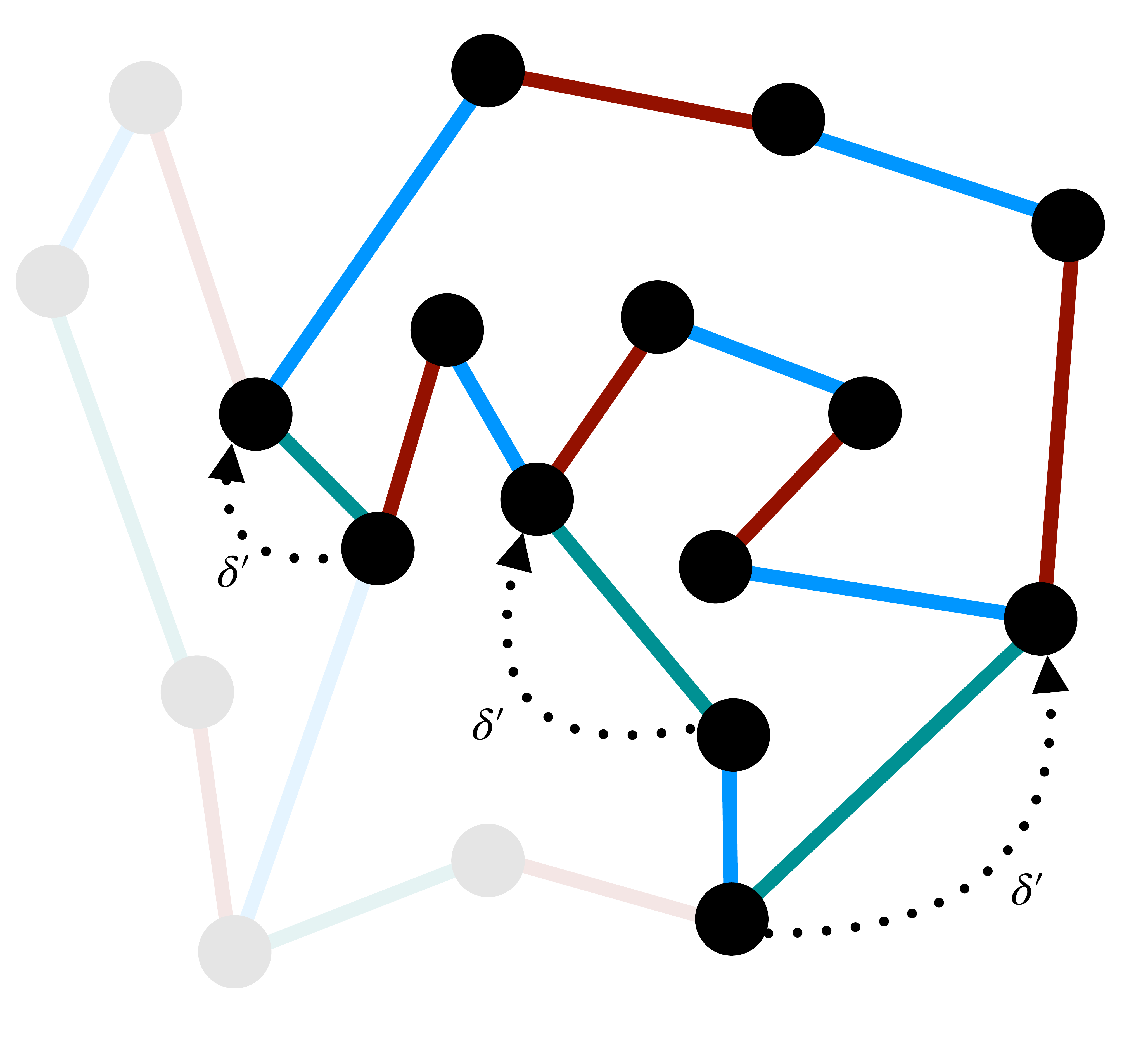}
        \caption{Demand routed.}\label{sfig:pgProof3}
    \end{subfigure} \hfill
    \begin{subfigure}[b]{0.24\textwidth}
        \centering
        \includegraphics[width=\textwidth,trim=0mm 0mm 260mm 0mm, clip]{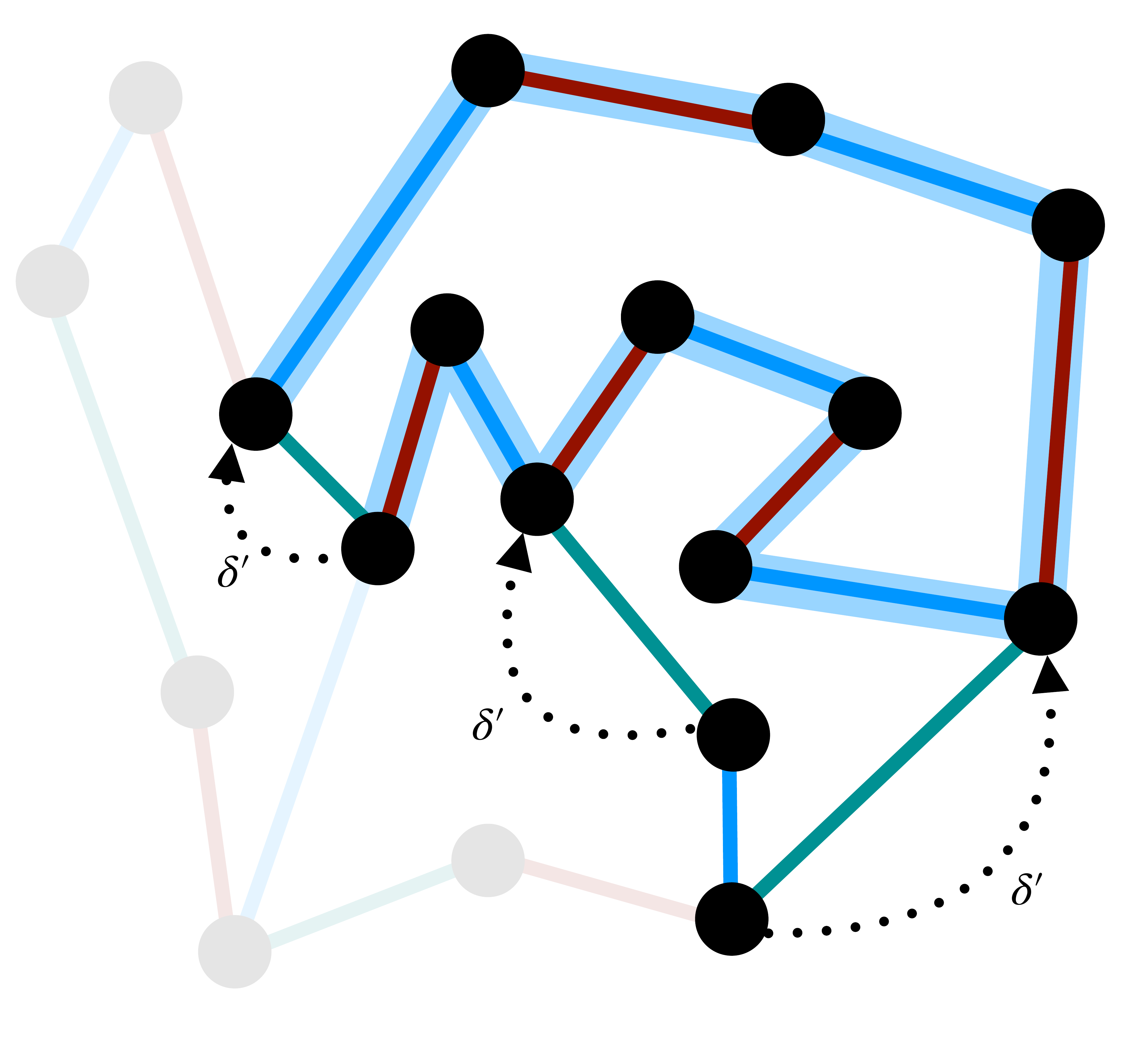}
        \caption{$t$-length path.}\label{sfig:pgProof4}
    \end{subfigure}
    \caption{Our proof of the bound on the arboricity of a $t$-\pg graph. \ref{sfig:pgProof1} gives our input $t$-\pg graph $G$ where the last matching $M$ is in green. \ref{sfig:pgProof2} gives the subgraph $H \subseteq G$ of large minimum degree implied by \Cref{lem:highMinDeg}. \ref{sfig:pgProof3} gives the demand that \Cref{lem: linked ED} routes with congestion $\delta'/2$ and \ref{sfig:pgProof4} gives in blue the $t$-length path between endpoints of an edge in $M$ that does not contain an edge in $M$. The existence of this path is implied by the flow of \Cref{lem: linked ED} as applied to $H$ and this path contradicts the fact that $G$ is $t$-\pg assuming $t \geq 10$.} \label{fig:pgProof}
\end{figure}

\gPGArb*
\begin{proof}

    Throughout this proof we use the following parameters:
$$\phi=1/(2\cdot t\cdot n^{\congslack/ t} \cdot \log n) \quad\quad\text{and}\quad\quad \delta=4\conConst\cdot t\log n/\phi^2 = t^3\cdot \log^3n \cdot n^{O(1/t)}$$%=(4\congslack^2)\cdot g^3\cdot n^{(2\hopslack/ g)}\cdot \log^2 n=n^{O(1/g)},$$
where $\congslack$ is the constant stated in \Cref{thm: linked ED} and $\conConst$ is the constant stated in \Cref{thm:flow character}.

Let $G$ be any $t$-\pg graph. Assume for contradiction that $G$ has arboricity greater than $2\delta$. It follows by \Cref{lem:highMinDeg} that $G$ has a non-empty (induced) subgraph $H$ of minimum degree at least $\delta$.

We now use \Cref{lem: linked ED} to complete the proof. We apply \Cref{lem: linked ED} to $H$, $\phi$ and $t$. Let $H'$ be the subgraph we obtain, so $H'$ is a subgraph of $H$ and a subgraph of $G$, and therefore also a $t$-\pg graph (by \Cref{lem: pg subgraph}).
Let $(E'_1,\ldots,E'_k)$ be a $t$-\pg sequence that generates $H'$. %Assume without loss of generality that each $E'_i$ in this sequence is a non-empty set. 
Denote $M=E'_k$ as the last set in this sequence.

%Consider the matching $M$. 
From \Cref{lem: linked ED}, there exists a flow $F$ that sends $\delta' = \frac{\phi}{2t} \cdot \delta$ units of flow from $u$ to $v$ in $H'$ for each edge $\{u,v\}\in M$ with congestion $\delta'/2$ and dilation $t$. Thus, the total amount of flow sent by $F$ from sources to sinks is $|M|\cdot \delta'$. On the other hand, each edge in $H'$ carries at most $\delta'/2$ units of flow and so all edges of $M$ carry at most $|M|\cdot (\delta'/2) < |M|\cdot \delta'$ units of flow.

Thus, some flow must be sent along a path that does not contain edges of $M$. Let $\{u,v\}$ be the endpoints of such a path. As the flow only uses paths of length at most $t$, this implies that $\dist_{H'\setminus M}(u,v)\le t$, and it follows that $\dist_{H'\setminus M}(u,v)\le t$, a contradiction to the fact that $H'$ is a $t$-\pg graph.
\end{proof}

\section{Sparse Spanner from Parallel Greedy (\Cref{thm:main} Proof)}
We conclude with our proof that the parallel greedy algorithm produces a sparse $t$-spanner.
\mainThm*
\begin{proof}
    The argument that the output graph $H$ is a $t$-spanner is identical to the usual analysis of the greedy algorithm; in particular, if $G$ has no $t$-unspanned edges with respect to $H$ then $H$ is a $t$-spanner. This follows since in such an $H$, every edge $\{u,v\} \in G$ is either in $H$ or $d_H(u,v) \leq t$. Thus, for two arbitrary vertices $u$ and $v$ with shortest path $P = (u = v_0, v_1, v_2, \ldots, v_k = v)$ in $G$ we have $d_H(v_{i-1}, v_{i}) \leq t$ for each $i \in [k]$ and so 
    \begin{align*}
        d_H(u,v) \leq \sum_{i \in [k]} d_H(v_{i-1}, v_{i}) \leq t \cdot |P| = t \cdot d_G(u,v).
    \end{align*}

    Likewise, by construction the output of the parallel greedy algorithm is a $t$-$\pg$ graph. Thus, by \Cref{thm: main} it has arboricity at most $t^3\cdot \log^3 n \cdot n^{O(1/t)}$ and so contains at most $t^3 \cdot \log^3 n \cdot n^{1 + O(1/t)}$ edges.
\end{proof}

\section{Conclusion and Future Directions}
In this work we showed that the classic greedy algorithm for constructing $t$-spanners produces sparse $t$-spanners even when many edges from a matching are added in parallel. Using a simple application of length-constrained expander decompositions, we demonstrated that this process results in a graph with at most $t^3 \cdot \log^3 n \cdot n^{1 + O(1/t)}$ edges.

The most obvious future direction is to tightly characterize the density of the spanners produced by this algorithm. In contrast, to the classic sequential greedy algorithm, simple examples (e.g.\ the $n$-dimensional hypercube) demonstrate that the $t$-spanners produced by the parallel greedy algorithm contain at least $\Omega(n \cdot \log n)$-many edges. As such, an exciting future direction is to determine whether or not this simple lower bounds can be matched on the upper bounds side. Furthermore, our work is a simple application of length-constrained expander decompositions. Work in algorithms work has only recently begun to fully develop applications of these decompositions and so we leave further applications as an exciting future direction. Likewise, we believe $t$-$\pg$ are a fundamental graph class and leave further applications of their sparsity as a future direction.

\section*{Acknowledgements}
We would like to thank Greg Bodwin for helpful comments and references.

Haeupler and Hershkowitz supported in part by NSF grants CCF-1527110, CCF-1618280, CCF-1814603, CCF-1910588, NSF CAREER award CCF1750808, a Sloan Research Fellowship, funding from the European Research Council (ERC) under the European Union’s Horizon 2020 research and innovation program (ERC grant agreement 949272) and the Swiss National Foundation (project grant 200021\_184735). Hershkowitz also funded by the SNSF, Swiss National Science Foundation grant 200021\_184622. Tan supported by a grant to DIMACS from the Simons Foundation (820931).
% \begin{center}
% \includegraphics[scale=0.1]{figures/SNF_logo_standard_office_color_pos_e.jpg}
% \end{center}

\bibliography{abb,main}
\bibliographystyle{alpha}

\end{document}